\documentclass[10pt,conference]{IEEEtran}
\IEEEoverridecommandlockouts

\usepackage{amsmath,amssymb,amsthm}
\usepackage{mathtools}
\usepackage{bm}
\usepackage{booktabs}
\usepackage{array}
\usepackage{multirow}
\usepackage{graphicx}
\usepackage{xcolor}
\usepackage{hyperref}
\usepackage{enumitem}
\usepackage{microtype}
\usepackage{cite}
\usepackage{comment}
\usepackage{subcaption}
\usepackage[most]{tcolorbox}

\tcbset{
  algobox/.style={
    enhanced,
    breakable,
    colback=gray!5,
    colframe=gray!55,
    boxrule=0.6pt,
    arc=3pt,
    left=8pt, right=8pt, top=6pt, bottom=6pt,
    fonttitle=\bfseries\small,
    attach boxed title to top left={yshift=-2mm, xshift=6pt},
    boxed title style={colback=gray!55, colframe=gray!55,
                       arc=2pt, boxrule=0pt,
                       left=4pt, right=4pt, top=2pt, bottom=2pt},
  }
}

\usepackage{tikz}
\usetikzlibrary{shapes.geometric, arrows.meta, positioning, calc, fit, backgrounds}
\usepackage{pgfplots}
\pgfplotsset{compat=1.18}
\usepackage{cleveref}

\hypersetup{
  colorlinks=true,
  linkcolor=blue!60!black,
  citecolor=green!50!black,
  urlcolor=blue!70!black
}

\theoremstyle{definition}
\newtheorem{definition}{Definition}

\theoremstyle{plain}
\newtheorem{theorem}{Theorem}
\newtheorem{lemma}[theorem]{Lemma}
\newtheorem{corollary}[theorem]{Corollary}
\theoremstyle{remark}
\newtheorem{remark}{Remark}

\newcommand{\denom}{\mathrm{D}}
\newcommand{\energyError}{\widetilde{E}}

\begin{document}

\title{Reachability Guarantees for Cart-Pole Swing-Up\\and Stabilization}

\author{
\IEEEauthorblockN{Mohamed Khalid~M Jaffar}
\IEEEauthorblockA{Department of Aerospace Engineering\\University of Maryland, College Park, USA\\Email: {\tt\small khalid26@umd.edu}}
\thanks{\indent *This work was not supported by any organization. The simulation scripts are available at \href{https://github.com/khalid2696/cartpole_swing_analysis.git}{\texttt{github.com/khalid2696/cartpole\_swing\_analysis}}.}
}  

\maketitle

\begin{abstract}
The cart-pole swing-up is a canonical benchmark for nonlinear control of underactuated systems, yet an end-to-end guarantee linking the global swing-up maneuver to the local stabilizer is seldom formalized. We present a reachability analysis of a switched energy-based/LQR controller that certifies convergence to the upright equilibrium from a compact set of initial conditions. The swing-up design exploits the phase-space geometry of the conservative pendulum: the upright equilibrium lies on the homoclinic orbit, and an energy-shaping law drives the energy error to zero, steering the pendulum onto this orbit; convergence follows from LaSalle's invariance principle. An augmented Lyapunov function additionally regulates the steady-state cart velocity to zero, and we prove almost-global convergence of the resulting closed-loop system. 
A local LQR with a certified ellipsoidal region of attraction stabilizes the upright equilibrium, and we verify numerically that the swing-up phase delivers the state into this region, formalizing the handoff. Numerical simulations corroborate the theoretical analysis.
\end{abstract}



\section{Introduction}

The cart-pole, or pendulum on a cart, is a classical benchmark for nonlinear control and planning of underactuated mechanical systems~\cite{tedrake2022,astrom2000}. It serves as a representative model for many practical systems, such as an overhead crane transporting a load, the pitch dynamics of a rocket with gimballed thrust, fuel slosh in a moving tank, and the balance models used to study bipedal locomotion~\cite{spong2020}. Its swing-up task: driving the pendulum from the stable hanging equilibrium to the unstable upright equilibrium while keeping the cart near the origin, requires traversal through a highly nonlinear region of the state space, making it a compelling testbed for formal verification.

Energy-shaping controllers are the dominant paradigm for swing-up of underactuated pendulums. \r{A}str\"om and Furuta~\cite{astrom2000} pump energy toward the upright homoclinic level set, whereas Chung and Hauser~\cite{chung1995} regulate the swing energy as an output. Spong's partial feedback linearization of an acrobot~\cite{spong1995} and a pendubot \cite{fantoni2000} extends the energy-based idea to two-link robots; Lozano et al.~\cite{lozano2000} stabilize the pendulum directly about its homoclinic orbit. 

A common challenge of energy-based controllers is ensuring internal stability --- the energy law governs the pendulum subspace but leaves the cart position and/or velocity unconstrained. Yang et al.~\cite{yang2009swing} address this by augmenting the Lyapunov candidate with a cart-velocity term, and Chatterjee et al.~\cite{chatterjee2002swing} enforce track limits through a potential-well construction. We build on the former, deriving a control law directly for the force input and providing what it leaves open: a complete invariant-set characterization, an instability proof for the downward equilibrium, and a closed-form cart-displacement bound. 

A related line of work certifies stability regions for controllers of locally-linearized systems: Lyapunov sublevel sets computed via sum-of-squares
programming~\cite{parrilo2000,tan2008} and the LQR-trees framework~\cite{tedrake2010}. These methods certify regions around equilibria or trajectories but typically assume the state is already delivered into the certified set. This paper supplies that missing link for the swing-up task in a cart-pole --- verifying that the energy-based controller deposits the state inside the LQR region of attraction, so that the swing-up maneuver and the local certificate compose into a single end-to-end guarantee. 

This work uses a switched energy-based/LQR controller designed from the phase-space geometry of the conservative pendulum: an energy-shaping law with a feedforward term drives the state onto the homoclinic orbit through the upright equilibrium. 
On this basis we make two contributions:
(i)~Building on the energy-plus-velocity Lyapunov construction of Yang et al.~\cite{yang2009swing}, we derive the control law directly for the force input and establish almost-global convergence properties of the resulting closed-loop system.
(ii)~We formalize the switching condition through a handoff region contained strictly within the LQR's certified region of attraction, yielding an end-to-end reachability guarantee from a compact set about the downward equilibrium (excluding it) to the upright equilibrium.

\section{System Description and Equations of Motion}
\label{sec:preliminaries}

The system comprises a cart of mass $M$ on a frictionless horizontal track
and a uniform pendulum of mass $m$ and length $\ell$ pivoting from the
cart.  The generalized coordinates are the cart position $x$ and the
pendulum angle $\theta$, measured from the \emph{downward} vertical
($\theta=0$ is the hanging equilibrium, $\theta=\pi$ is the upright
equilibrium). The control input $u$ is the horizontal force $F$.  The state
vector is $\mathbf{x} = [x,\,\dot{x},\,\theta,\,\dot{\theta}]^{\top}$. Introducing shorthand,
\begin{equation}
    s_\theta \triangleq \sin\theta, \: c_\theta \triangleq
\cos\theta, \: \text{and} \: \denom \triangleq M + ms_\theta^2,
\end{equation}
the Euler--Lagrange equations of cart-pole dynamics are:
\begin{equation}
\small
  \dot{\mathbf{x}} = \mathbf{f}(\mathbf{x},u) \triangleq
  \begin{bmatrix}
    \dot{x} \\[4pt]
    \dfrac{F + m\ell\dot{\theta}^2 s_\theta + mgs_\theta c_\theta}{\denom}
    \\[10pt]
    \dot{\theta} \\[4pt]
    \dfrac{-Fc_\theta - m\ell\dot{\theta}^2 s_\theta c_\theta
           - (M+m)gs_\theta}{\ell\denom}
  \end{bmatrix}.
  \label{eq:eom}
\end{equation}
The system has two equilibria of interest: the downward equilibrium
$\mathbf{x}_{\downarrow}=[0,0,0,0]^{\top}$ (open-loop stable) and the
upright equilibrium $\mathbf{x}^*=[0,0,\pi,0]^{\top}$ (open-loop
unstable).

\section{Energy-Based Swing-Up and LQR Stabilization}
\label{sec:approach}


The control objective is a \emph{swing-up and stabilization} task: starting from the open-loop stable downward equilibrium, bring the pendulum to the unstable upright equilibrium and hold it there, while keeping the cart near the origin. Because the upright equilibrium is
unstable and the system is underactuated (one actuated degree of freedom, two mechanical degrees of freedom), no single linear controller can achieve this globally. 

Our proposed approach decomposes the control task into two phases connected by a switching rule: 1) a \textit{swing-up phase}, in which an energy-shaping controller pumps energy into the pendulum until its trajectory approaches the energy level set containing the upright equilibrium, and 2) a \textit{stabilization phase}, in which an LQR controller with a certified ellipsoidal region of attraction (RoA) takes over once the state enters a designed switching region contained in that RoA. We first build the intuition underlying the energy-based design.

\subsection{Energy-Based Swing-Up Strategy}

\paragraph{Invariant orbits of the conservative pendulum}
Whenever the cart does not accelerate ($\ddot{x} = 0$), the pendulum dynamics in \eqref{eq:eom} decouple from the cart and reduce to the conservative pendulum equation
\begin{equation}
  m\ell^2\ddot{\theta} + mg\ell\sin\theta = 0,
  \label{eq:conservative_pendulum}
\end{equation}
whose total mechanical energy
\begin{equation}
  E(\theta,\dot{\theta}) = \tfrac{1}{2}m\ell^2\dot{\theta}^2
    - mg\ell\cos\theta
  \label{eq:energy}
\end{equation}
is a first integral. With $E^* = mg\ell$ denoting the energy at the upright equilibrium, define the energy error
\begin{equation}
  \energyError(\theta,\dot{\theta}) \triangleq E - E^*
    = \tfrac{1}{2}m\ell^2\dot{\theta}^2 - mg\ell(1+\cos\theta).
  \label{eq:energy_error}
\end{equation}
Under \eqref{eq:conservative_pendulum}, every level set
$\mathcal{S}_c \triangleq \{(\theta,\dot{\theta}) \:|\: \energyError = c\}$ is an
invariant orbit (see Fig.~\ref{fig:orbits}) --- periodic oscillations about the downward equilibrium for $-2mg\ell < c < 0$, and full rotations for $c > 0$. Separating the two modes is the level set $\mathcal{S}_0$, the homoclinic orbit through
the upright equilibrium $(\pi, 0)$, on which
\begin{equation}
  \dot{\theta}^2 = \frac{2g}{\ell}(1+\cos\theta)
    = \frac{4g}{\ell}\cos^2\frac{\theta}{2}.
  \label{eq:homoclinic}
\end{equation}

The pendulum state $(\theta,\dot{\theta})$ evolves on the cylinder $\mathbb{S}^1 \times \mathbb{R}$. 
The target state $(\pi,0)$ lies on the homoclinic orbit $\mathcal{S}_0$, so a swing-up controller should (i)~drive $\energyError \to 0$, and (ii)~degenerate to a pure feedforward on $\mathcal{S}_0$ that renders $\ddot{x} = 0$, so that $\mathcal{S}_0$ remains invariant for the closed loop. The homoclinic orbit then becomes an attractive limit set: the pendulum converges to $\mathcal{S}_0$ and travels along it, repeatedly passing arbitrarily close to the upright equilibrium, where a local stabilizer completes the control task. Both the following swing-up control laws are built on this principle.

\begin{figure}[t]
\centering
\includegraphics[trim={0cm 0.45cm 0cm 0cm}, clip=true, width=1\columnwidth]{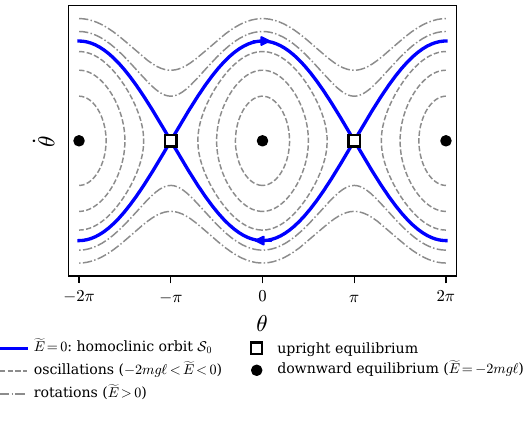}
\caption{Phase portrait of the decoupled pendulum ($\ddot{x}=0$). Each energy level set $\mathcal{S}_c = \{(\theta,\dot{\theta}) \:|\: \energyError = c\}$ is an invariant orbit; the
homoclinic orbit $\mathcal{S}_0$ ($\widetilde{E}=0$, solid blue) passes through
the upright equilibria (squares) and separates oscillations (dashed) from rotations (dash-dotted), while filled circles represent the downward equilibrium. The swing-up laws \eqref{eq:swing_control_law} and \eqref{eq:swing_control_law_aug} render $\mathcal{S}_0$ an attractive limit set of the closed loop system.}
\label{fig:orbits}
\end{figure}

\paragraph{Energy Lyapunov function and control law synthesis}
Following the energy-shaping approach of~\cite{spong1995,astrom2000,spong2020}, propose
\begin{equation}
  {V(\theta,\dot{\theta}) = \tfrac{1}{2}\energyError^2}.
  \label{eq:Lyapunov_function}
\end{equation}
$V$ is positive semi-definite, equals zero on $\{\energyError=0\}$, and is radially unbounded in the energy error. Differentiating \eqref{eq:Lyapunov_function} along solutions of \eqref{eq:eom}: $\dot{V} = \energyError\dot{E}$.
Differentiating \eqref{eq:energy} and substituting $\ddot{\theta}$ from
\eqref{eq:eom}:
\begin{equation} \small
  \dot{E}
  = m\ell\dot{\theta}\cdot
    \frac{-Fc_\theta - m\ell\dot{\theta}^2 s_\theta c_\theta
          - (M+m)gs_\theta}{\denom}
    + mg\ell s_\theta\dot{\theta}.
  \label{eq:Edot_sub}
\end{equation}
Collecting the gravity terms and using
$\denom-(M+m) = ms_\theta^2-m = -mc_\theta^2$:
\begin{equation}
  \frac{m\ell\dot{\theta}}{\denom}\bigl[-(M+m)gs_\theta\bigr]
  + mg\ell s_\theta\dot{\theta}
  = \frac{-m^2\ell\dot{\theta}gs_\theta c_\theta^2}{\denom}.
  \label{eq:grav_cancel}
\end{equation}
Substituting \eqref{eq:grav_cancel} back into \eqref{eq:Edot_sub}:
\begin{equation}
  \dot{E}
  = \frac{-m\ell\dot{\theta}c_\theta}{\denom}\,F
    - \frac{m^2\ell\dot{\theta}s_\theta c_\theta
            (\ell\dot{\theta}^2 + gc_\theta)}{\denom}.
  \label{eq:Edot_full}
\end{equation}
The first term is the power delivered to the pendulum by $F$; the second
is the autonomous conservative exchange internal to the pendulum subsystem,
denoted $\Phi$.  Hence
\begin{align}
\begin{split}
  \dot{V} &= \energyError\left[
    \frac{-m\ell\dot{\theta}c_\theta}{\denom}\,F + \Phi
  \right] \\
  &= \energyError \cdot \frac{-m\ell\dot{\theta}c_\theta}{\denom} \left[
    F + m s_\theta(\ell\dot{\theta}^2 + gc_\theta) \right].
  \label{eq:Vdot_split}
\end{split}
\end{align}

To render $\dot{V}\leq 0$, cancel the autonomous term $\Phi$ via feedforward and inject a negative definite contribution. Choosing the term in bracket as, $F + m s_\theta(\ell\dot{\theta}^2 + gc_\theta) = k\,\energyError\,\dot{\theta} c_\theta$,
and solving for $F$:
\begin{equation}
  {F = k\,\energyError\,\dot{\theta}\,c_\theta
           - m s_\theta(\ell\dot{\theta}^2 + gc_\theta),
           \qquad k>0.}
  \label{eq:swing_control_law}
\end{equation}
The first term shapes the energy; the second cancels the autonomous term, ensuring no cross-term appears in $\dot{V}$.  Substituting \eqref{eq:swing_control_law} into \eqref{eq:Vdot_split} and noting that $\denom > 0$:
\begin{equation}
  \dot{V} = - \frac{m\ell k\,\energyError^2\dot{\theta}^2 c_\theta^2}{\denom} \leq 0.
  \label{eq:Vdot}
\end{equation}
Hence $\dot{V}\leq 0$ everywhere.


\begin{remark}
  The feedforward term $-ms_\theta(\ell\dot{\theta}^2+gc_\theta)$ in
  \eqref{eq:swing_control_law} is absent from the standard energy-pumping law of \cite{spong1995,astrom2000}, where the autonomous term $\Phi$ is left in $\dot{V}$ and argued to vanish asymptotically. Including it gives a cleaner $\dot{V}$ expression \eqref{eq:Vdot}, removes any residual coupling between the pendulum and cart dynamics in the Lyapunov analysis, and improves closed-loop performance in simulation.
\end{remark}

\begin{theorem}[\textsc{Swing-up convergence}]\label{thm:swingup}
Under control law \eqref{eq:swing_control_law}, every trajectory of
\eqref{eq:eom} with
$(\theta_0,\dot{\theta}_0) \neq (0,0)$ satisfies $\energyError(t) \to 0$ as
$t \to \infty$; that is, the pendulum state converges to the homoclinic orbit $\mathcal{S}_0$.
\end{theorem}

\begin{proof}
By \eqref{eq:Vdot}, $V \geq 0$ and $\dot{V} \leq 0$, so $V(t)$ is non-increasing and $|\energyError(t)| \leq |\energyError(0)|$ for all $t$; sublevel sets of $V$ are compact on the cylinder and positively invariant. By LaSalle's invariance principle~\cite{Khalil2002}, the
$\omega$-limit set of every trajectory lies in the largest invariant subset $\mathcal{M}$ of
$\mathcal{Z} \triangleq \{\dot{V} = 0\}
  = \{\energyError = 0\} \cup \{\dot{\theta} = 0\} \cup \{c_\theta = 0\}$. We now characterize $\mathcal{M}$.

On $\{\energyError = 0\}$, the energy-shaping term of
\eqref{eq:swing_control_law} vanishes and
$F = -ms_\theta(\ell\dot{\theta}^2 + gc_\theta)$, which upon substitution into \eqref{eq:eom} gives $\ddot{x} = 0$; the pendulum evolves conservatively and $\mathcal{S}_0$ is invariant. A trajectory with $\dot{\theta} \equiv 0$ requires $\ddot{\theta} = -(g/\ell)s_\theta \equiv 0$, leaving only the equilibria $(0,0)$ and $(\pi,0) \in \mathcal{S}_0$. The set $\{c_\theta = 0\}$ contains no invariant subset, since $\ddot{\theta} \neq 0$ at
$(\pm\pi/2, 0)$. Hence $\mathcal{M} = \mathcal{S}_0 \cup \{(0,0)\}$.
It remains to exclude the downward equilibrium $(0,0)$, where $\energyError = -2mg\ell$.

\emph{Case A} ($\energyError(0) < 0$): since $\energyError = -2mg\ell$ holds only at $(0,0)$, the assumption
$(\theta_0,\dot{\theta}_0) \neq (0,0)$ gives
$|\energyError(0)| < 2mg\ell$, hence
$|\energyError(t)| \leq |\energyError(0)| < 2mg\ell$ for all $t$, and the trajectory cannot reach $(0,0)$.

\emph{Case B} ($\energyError(0) \geq 0$): then $\energyError(t) \geq 0$
for all $t$. Indeed, since $\energyError(t)$ is continuous, any transition to a negative value would require $\energyError(t_1) = 0$ at some $t_1$, placing the state on the invariant set $\mathcal{S}_0$, on which it remains thereafter --- preventing obtaining any negative $\energyError$ value, so the trajectory cannot reach $(0,0)$.

In either case the $\omega$-limit set excludes $(0,0)$ and is therefore contained in $\mathcal{S}_0$, giving $\energyError(t) \to 0$.
\end{proof}

\begin{corollary}[\textsc{Finite-time arrival near the upright
equilibrium}]\label{cor:arrival}
For any $\delta_\theta, \delta_\omega > 0$, every trajectory as in
Theorem~\ref{thm:swingup} enters the set
$\{|\theta - \pi| \leq \delta_\theta,\ |\dot{\theta}| \leq
\delta_\omega\}$ in finite time.
\end{corollary}

\begin{proof}
Since $|\energyError(t)|$ is non-increasing with $\energyError(t)\to 0$, choose $\varepsilon\in(0,\,2mg\ell)$ and $T_0$ such that $|\energyError(t)|\leq\varepsilon$ for all $t\geq T_0$. By the sign-preservation argument in the proof of Theorem~\ref{thm:swingup} (a change of sign of $\energyError$ would require crossing the invariant set $\mathcal{S}_0$), the sign of $\energyError$ is fixed for $t\geq T_0$, giving three cases.

\emph{Oscillatory case ($\energyError<0$)}: With $|\energyError|<2mg\ell$ the motion is a periodic oscillation, so a turning point $\theta_t$ (where $\dot\theta=0$) recurs within one half-period and is therefore reached in finite time. There $1+\cos\theta_t = |\energyError|/(mg\ell) \leq \varepsilon/(mg\ell)$; writing $\theta_t = \pi-\alpha$ gives $1-\cos\alpha \leq \varepsilon/(mg\ell)$, hence $|\theta_t-\pi| = \mathcal{O}(\sqrt{\varepsilon})$. At this instant both target conditions hold: $\dot\theta = 0 \leq \delta_\omega$ and, for $\varepsilon$ small enough, $|\theta_t-\pi|\leq\delta_\theta$. 

\emph{Rotational case ($\energyError>0$)}: Here $\dot\theta$ never
vanishes, so it keeps its sign and $\theta$ advances monotonically,
crossing $\theta=\pi$ in finite time. At the crossing, $\cos\theta=-1$, so $\energyError = \tfrac{1}{2}m\ell^2\dot\theta^2$ and $|\dot\theta| = \sqrt{2\energyError/(m\ell^2)} \leq
\sqrt{2\varepsilon/(m\ell^2)}$; thus $|\theta-\pi| = 0 \leq \delta_\theta$ and, for $\varepsilon$ small enough, $|\dot\theta|\leq\delta_\omega$.

\emph{Separatrix case ($\energyError\equiv 0$)}: The state evolves along $\mathcal{S}_0$ toward the saddle $(\pi,0)$. Since the target set is a neighborhood of $(\pi,0)$ with $\delta_\theta,\delta_\omega>0$, the trajectory enters it in finite time, even though the saddle itself is approached only as $t\to\infty$.

In every case, choosing $\varepsilon>0$ small enough that both estimates fall within $\delta_\theta$ and $\delta_\omega$ ensures the trajectory enters $\{|\theta-\pi|\leq\delta_\theta,\ |\dot\theta|\leq\delta_\omega\}$ in finite time.
\end{proof}

\paragraph{Augmented Lyapunov function for cart-velocity regulation}
The Lyapunov function \eqref{eq:Lyapunov_function} governs only the pendulum subsystem; it leaves the cart velocity $\dot{x}$ unconstrained during swing-up (indeed, on $\mathcal{S}_0$ the feedforward term makes $\ddot{x}=0$, so the cart coasts at whatever velocity it has accrued). To additionally drive $\dot{x}\to 0$, we follow the energy-plus-velocity Lyapunov function construction of Yang et al.~\cite{yang2009swing} and augment the candidate with the cart-velocity term,
\begin{equation}
  V_{\mathrm{aug}}(\dot{x},\theta,\dot{\theta})
    = \tfrac{1}{2}\energyError^2 + \tfrac{\mu}{2}\dot{x}^2,
    \quad \mu > 0.
  \label{eq:Lyapunov_function_aug}
\end{equation}
Differentiating along \eqref{eq:eom} and factoring (Appendix~\ref{app:aug}) yields
\begin{equation}
  \dot{V}_{\mathrm{aug}}
    = \frac{\sigma}{\denom}
      \bigl[\, F + m s_\theta(\ell\dot{\theta}^2 + gc_\theta)\,\bigr],
    \quad
    \sigma \triangleq \mu\dot{x} - m\ell\dot{\theta}c_\theta\,\energyError.
  \label{eq:Vdot_aug_factored}
\end{equation}
Choosing the bracket equal to $-k\sigma$ gives the globally defined control law
\begin{equation}
  F_{\mathrm{aug}} = -k\sigma
    - m s_\theta(\ell\dot{\theta}^2 + gc_\theta),
    \qquad k>0,
  \label{eq:swing_control_law_aug}
\end{equation}
under which
\begin{equation}
  \dot{V}_{\mathrm{aug}} = -\frac{k\sigma^2}{\denom} \leq 0.
  \label{eq:Vdot_aug}
\end{equation}
Note that setting $\mu = 0$ in \eqref{eq:swing_control_law_aug} recovers the
unaugmented control law \eqref{eq:swing_control_law}.

\begin{theorem}[\textsc{Swing-up convergence under $V_{\mathrm{aug}}$}]
\label{thm:swingup_aug}
Under control law \eqref{eq:swing_control_law_aug} with $k,\mu>0$, every trajectory of \eqref{eq:eom} whose initial condition does not lie on the stable manifold of the downward equilibrium set $\mathcal{M}_B \triangleq \{(x,0,0,0) \:|\: x \in \mathbb{R}\}$ satisfies $\dot{x}(t)\to 0$ and
$\energyError(t)\to 0$ as $t\to\infty$; that is, the cart comes to rest
and $(\theta,\dot{\theta})$ converges to the homoclinic orbit
$\mathcal{S}_0$.
\end{theorem}

\begin{proof}
$V_{\mathrm{aug}}\geq 0$ and, by \eqref{eq:Vdot_aug},
$\dot{V}_{\mathrm{aug}}\leq 0$, so every sublevel set
$\{V_{\mathrm{aug}}\leq c\}$ is compact and positively invariant. By
LaSalle's invariance principle~\cite{Khalil2002}, every trajectory
converges to the largest invariant set $\mathcal{M}$ in
$\{\dot{V}_{\mathrm{aug}}=0\}=\{\sigma=0\}$. Appendix~\ref{app:aug}
characterizes this set as
$\mathcal{M} = \mathcal{M}_A \cup \mathcal{M}_B$, where $\mathcal{M}_A = \bigl\{(x,\dot{x},\theta,\dot{\theta}) \:\big|\:
    \dot{x}=0,\ (\theta,\dot{\theta})\in\mathcal{S}_0 \bigr\}$ is
the cart-at-rest homoclinic orbit \eqref{eq:MA} and
$\mathcal{M}_B = \{(x,0,0,0)\}$ is the downward equilibrium set, and shows
$\mathcal{M}_B$ is unstable --- the closed-loop Jacobian there has
eigenvalues with positive real part in the $(\theta,\dot{\theta})$
directions. Hence every trajectory not initialized on the stable manifold
of $\mathcal{M}_B$ converges to $\mathcal{M}_A$: the cart comes to rest
($\dot{x}\to 0$) and $(\theta,\dot{\theta})\to\mathcal{S}_0$
($\energyError\to 0$).
\end{proof}

\begin{remark}
As in Corollary~\ref{cor:arrival}, convergence to $\mathcal{S}_0$ with $\dot{x}\to 0$ implies the state enters any neighborhood of the upright equilibrium in finite time; the augmented law additionally ensures $\dot{x}(T_{\mathrm{sw}})\approx 0$, so the cart-velocity contribution to the switching state is negligible --- easing the containment condition of local LQR's region of attraction.
\end{remark}

While $V_{\mathrm{aug}}$ regulates the cart \emph{velocity}, the cart \emph{position} $x$ remains unactuated in the Lyapunov analysis and may drift. The following bound shows the drift is nonetheless finite over the swing-up horizon, so the switching state cannot be arbitrarily far from
the origin.

\begin{lemma}[\textsc{Cart-drift bound under control law
\eqref{eq:swing_control_law_aug}}]\label{lem:cart_drift}
Starting from $x(0)=0$, $\dot{x}(0)=0$, and initial energy error $\energyError_0$, the cart states under \eqref{eq:swing_control_law_aug} satisfy
\begin{equation}
  |\dot{x}(t)| \leq \bar{v}
    \triangleq \frac{|\energyError_0|}{\sqrt{\mu}},
  \qquad
  |x(t)| \leq \bar{v}\,t \leq \bar{x}
    \triangleq \bar{v}\,T_{\mathrm{sw}},
  \label{eq:cart_bounds}
\end{equation}
for all $t\in[0,T_{\mathrm{sw}}]$ where $T_{\mathrm{sw}}$ denotes a finite switch time.
\end{lemma}

\begin{proof}
By \eqref{eq:Vdot_aug}, $V_{\mathrm{aug}}$ is non-increasing, so for all
$t$,
$\tfrac{\mu}{2}\dot{x}(t)^2 \leq V_{\mathrm{aug}}(t)
  \leq V_{\mathrm{aug}}(0) = \tfrac{1}{2}\energyError_0^2$,
using $\dot{x}(0)=0$. Hence
$|\dot{x}(t)| \leq |\energyError_0|/\sqrt{\mu} = \bar{v}$. Integrating
from $x(0)=0$ and applying the triangle inequality for integrals gives
$|x(t)| \leq \int_0^t |\dot{x}(\tau)|\,d\tau \leq \bar{v}\,t \leq
\bar{v}\,T_{\mathrm{sw}} = \bar{x}$.
\end{proof}

The velocity bound $\bar{v}$ is derived directly from the Lyapunov decrease and is independent of any trajectory integration. It is directly adjustable through the control parameter $\mu$ --- increasing $\mu$ tightens $\bar{v}$ at the cost of a stiffer control law.

\begin{remark}\label{rmk:conservative_bound}
The position bound $\bar{x}$ is valid but conservative, since it treats $|\dot{x}(t)|$ as fixed at its maximum throughout swing-up, whereas the $\tfrac{\mu}{2}\dot{x}^2$ term actively damps $\dot{x}$ toward zero, rendering the cart velocity oscillatory and mean-reverting. From numerical experiments (Section~\ref{sec:results}), the observed excursion is far smaller than
$\bar{x}$, and the gains $k,\mu$ can be tuned to keep the cart near the
origin at switching.
\end{remark}

\begin{remark}[\textsc{Hard position constraints via a potential well}]
\label{rmk:potential_well}
Lemma~\ref{lem:cart_drift} bounds the cart excursion but does not \emph{enforce} a prescribed track limit $|x|\leq x_{\mathrm{bound}}$. Where such a hard constraint is required, an additive displacement-penalizing term such as the logarithmic ``potential well''
of Chatterjee et al.~\cite{chatterjee2002swing},
$u_{\mathrm{well}} = k_{\mathrm{cw}}\,\mathrm{sgn}(x)\,
  \log\!\bigl(1 - |x|/x_{\mathrm{bound}}\bigr)$ with $k_{\mathrm{cw}} > 0$,
repels the cart from the boundary as $|x|\to x_{\mathrm{bound}}$.
Incorporating such a term into $F_{\mathrm{aug}}$ while retaining a clean $\dot{V}\leq 0$ certificate is nontrivial and left to future work; here we rely on the analytical finite bound \eqref{eq:cart_bounds} and gain tuning.
\end{remark}

\subsection{LQR Stabilization and Ellipsoidal Region of Attraction}
\label{subsec:lqr}

The energy-based laws of the previous subsection drive
$(\theta,\dot{\theta})$ onto the homoclinic orbit $\mathcal{S}_0$ and, under \eqref{eq:swing_control_law_aug}, the cart velocity to zero, while Lemma~\ref{lem:cart_drift} certifies the cart position remains within a finite excursion off the origin. By Corollary~\ref{cor:arrival}, the state therefore enters any prescribed neighborhood of the upright equilibrium
in finite time, with all four components bounded. What the energy laws do not provide is asymptotic stabilization \emph{at} the upright equilibrium: convergence along $\mathcal{S}_0$ is only to the orbit, and $(\pi,0)$ is approached but never reached in finite time. This motivates handing off to a local controller that renders the upright equilibrium asymptotically stable. 

We use an LQR designed via local linearization at the upright state, and certify an ellipsoidal region of attraction (RoA) into which the swing-up phase deposits the system state.

\paragraph{Linearization at the upright equilibrium}
Let $\phi = \theta - \pi$ and $\mathbf{z} = [x,\dot{x},\phi,\dot{\phi}]^\top$. For small $\phi$, $s_\theta \approx -\phi$, $c_\theta \approx -1$, and
$\denom \approx M$. Linearizing \eqref{eq:eom} about $\mathbf{z}^* = 0$,
\begin{equation}
  \dot{\mathbf{z}} = A\mathbf{z} + Bu,
  \label{eq:linearized}
\end{equation}
\begin{equation}
  A = \begin{bmatrix}
    0 & 1 & 0 & 0 \\
    0 & 0 & \frac{mg}{M} & 0 \\
    0 & 0 & 0 & 1 \\
    0 & 0 & \frac{(M+m)g}{M\ell} & 0
  \end{bmatrix},
  \qquad
  B = \begin{bmatrix} 0 \\ \frac{1}{M} \\ 0 \\ \frac{1}{M\ell} \end{bmatrix}.
  \label{eq:AB}
\end{equation}
The pair $(A,B)$ is controllable, and $A$ has a single unstable
eigenvalue $\lambda = +\sqrt{(M+m)g/(M\ell)}$.

\paragraph{LQR design}
For weights $Q \succeq 0$, $R > 0$, the optimal gain is
\begin{equation}
  K = R^{-1}B^\top P_{\mathrm{lqr}},
  \label{eq:lqr_gain}
\end{equation}
where $P_{\mathrm{lqr}} \succ 0$ solves the algebraic Riccati equation (ARE)
\begin{equation}
  A^\top P_{\mathrm{lqr}} + P_{\mathrm{lqr}}A
    - P_{\mathrm{lqr}}BR^{-1}B^\top P_{\mathrm{lqr}} + Q = 0,
  \label{eq:are}
\end{equation}
and local control law
\begin{equation}
    u_{\mathrm{lqr}}=-K\mathbf{z}
    \label{eq:LQR_control_law}
\end{equation}
asymptotically stabilizes the linearized system~\eqref{eq:linearized}, driving arbitrarily small $\mathbf{z} \to \mathbf{z}^*=\mathbf{0}$.

\paragraph{Region of attraction for the nonlinear system}
Consider the LQR value function as a Lyapunov candidate for the closed-loop
\emph{nonlinear} system,
\begin{equation}
  W(\mathbf{z}) = \mathbf{z}^\top P_{\mathrm{lqr}} \mathbf{z}.
  \label{eq:LQR_Lyapunov_function}
\end{equation}
Writing the closed loop dynamics under \eqref{eq:LQR_control_law} as
$\dot{\mathbf{z}} = (A-BK)\mathbf{z} + g(\mathbf{z})$ with $g(\mathbf{z})\sim\mathcal{O}(\|\mathbf{z}\|^2)$ denoting the higher-order terms, and using
$\dot{W}|_{\mathrm{lin}} = -\mathbf{z}^\top Q_e \mathbf{z}$ with
$Q_e = Q + K^\top R K \succ 0$,
\begin{equation}
  \dot{W}_{\mathrm{nl}} = -\mathbf{z}^\top Q_e \mathbf{z} + 2\mathbf{z}^\top P_{\mathrm{lqr}}\, g(\mathbf{z}).
  \label{eq:Wdot_nl}
\end{equation}

\begin{definition}[\textsc{Certified RoA of the LQR controller}]
\label{def:roa}
For the ellipsoidal sublevel set
$\Omega_c \triangleq \{\mathbf{z}\in\mathbb{R}^4 \:|\: \mathbf{z}^\top P_{\mathrm{lqr}} \mathbf{z}
  \leq c\}$,
let critical level
\begin{equation}
  c^* \triangleq \min_{\mathbf{z}:\,\dot{W}_{\mathrm{nl}}=0,\, \mathbf{z}\neq 0} W(\mathbf{z}).
  \label{eq:cstar}
\end{equation}
\end{definition}

\begin{lemma}[\textsc{LQR local asymptotic stability}]\label{lem:lqr}
Let $c \leq c^*$. For all $\mathbf{z}(0)\in\Omega_c$, the closed-loop trajectory under $u=-K\mathbf{z}$ \eqref{eq:LQR_control_law} satisfies $\mathbf{z}(t)\in\Omega_c$ for all $t\geq 0$ and
$\mathbf{z}(t)\to 0$ as $t\to\infty$; that is,
$[x,\dot{x},\theta,\dot{\theta}]^\top \to [0,0,\pi,0]^\top$
\end{lemma}

\begin{proof}
Since $c \leq c^*$, \eqref{eq:cstar} gives $\dot{W}_{\mathrm{nl}} < 0$ on $\Omega_c\setminus\{\mathbf{0}\}$. Forward invariance follows: if
$\mathbf{z}(0)\in\Omega_c$ then $W(\mathbf{z}(t)) \leq W(\mathbf{z}(0)) \leq c$ for all $t\geq 0$.
As $W$ is strictly decreasing and bounded below by zero,
$W(t)\to W_\infty \geq 0$; if $W_\infty > 0$ the trajectory would remain in a compact set away from the origin where
$\dot{W}_{\mathrm{nl}} \leq -\epsilon < 0$ for some positive $\epsilon$, contradicting convergence of $W$. Hence $W_\infty = 0$ and $\mathbf{z}(t)\to 0$.
\end{proof}

\subsection{Overall Switching Control Architecture (see \Cref{fig:flowchart})}
\label{subsec:switching}

The two phases are joined by a switching condition. Define the switching region
in the upright-error coordinates $\mathbf{z}=[x,\dot{x},\phi,\dot{\phi}]^\top$,
\begingroup
\fontsize{9.5}{11.5}\selectfont
\begin{equation}
  \Sigma \triangleq \bigl\{\mathbf{z}\in\mathbb{R}^4 :
    |x|\leq\delta_x,\ |\dot{x}|\leq\delta_v,\
    |\phi|\leq\delta_\phi,\ |\dot{\phi}|\leq\delta_\omega \bigr\},
  \label{eq:switching_region}
\end{equation}
\endgroup
for thresholds $\delta_x,\delta_v,\delta_\phi,\delta_\omega > 0$. The
combined controller is
\begin{equation}
  u = \begin{cases}
    F_{\mathrm{aug}} \text{ from \eqref{eq:swing_control_law_aug}}, & \mathbf{z}\notin\Sigma,\\[2pt]
    -K\mathbf{z} \text{ from \eqref{eq:LQR_control_law}}, & \mathbf{z}\in\Sigma.
  \end{cases}
  \label{eq:switching_controller}
\end{equation}
Switching is one-way: once $\mathbf{z}\in\Sigma$ the LQR is engaged and not
disengaged, since Lemma~\ref{lem:lqr} guarantees forward invariance of
$\Omega_{c^*}\supseteq\Sigma$.

For the handoff to be formally valid, the switching region must lie inside
the certified RoA, $\Sigma\subseteq\Omega_{c^*}$. Unlike $\Sigma$, which is
a box in the four state components, $\Omega_{c^*}$ is an ellipsoid coupling
them through $P_{\mathrm{lqr}}$; containment is therefore a joint condition
on the thresholds.

\begin{remark}[\textsc{Sufficient condition for containment}]
\label{rmk:containment}
By the Rayleigh--Ritz inequality,
$\sup_{\mathbf{z}\in\Sigma} W(\mathbf{z})
  \leq \lambda_{\max}(P_{\mathrm{lqr}})
       (\delta_x^2 + \delta_v^2 + \delta_\phi^2 + \delta_\omega^2)$,
so $\Sigma\subseteq\Omega_{c^*}$ holds whenever the right-hand side is at
most $c^*$. This bound is convenient but conservative; a tight check
evaluates $\sup_{\mathbf{z}\in\Sigma}W(\mathbf{z})$ directly, or verifies
$W(\mathbf{z}(T_{\mathrm{sw}}))\leq c^*$ on the realized trajectory (see numerical results in Section~\ref{sec:results}).
\end{remark}

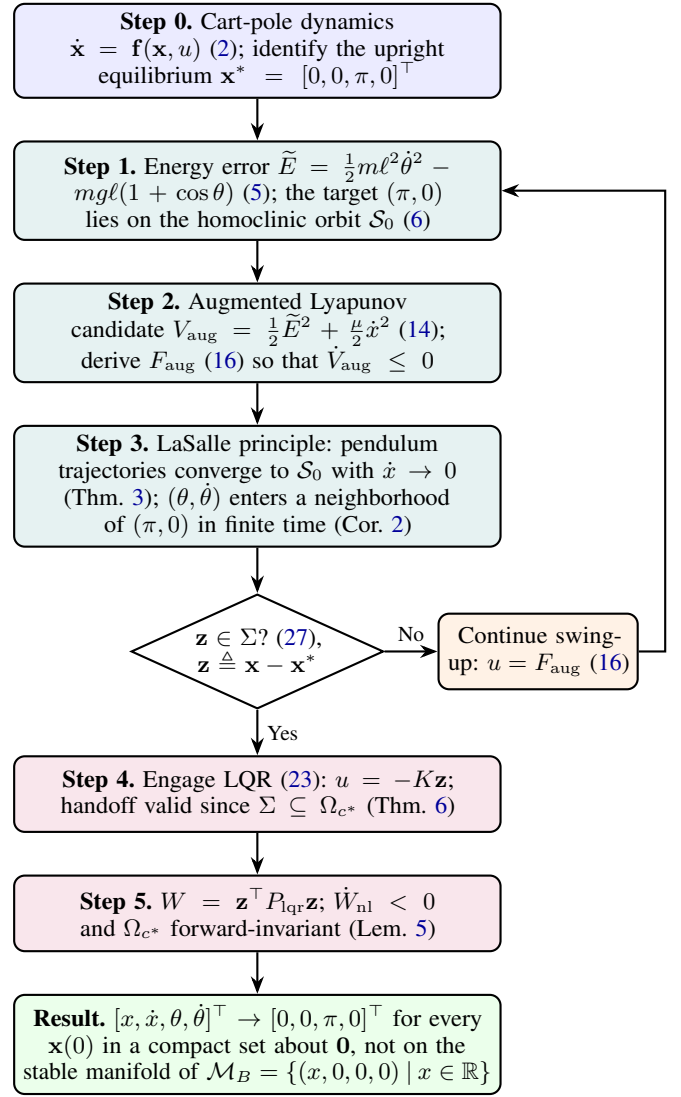
\begin{figure}[t]
\centering
\begin{tikzpicture}[
    >=Stealth,
    every node/.style={font=\small},
    block/.style={rectangle, rounded corners=4pt, draw, thick,
                  text width=6.2cm, minimum height=1.0cm,
                  align=center, fill=white},
    decision/.style={diamond, draw, thick, aspect=2.2,
                     text width=1.8cm, align=center, fill=white, inner sep=1pt},
    arrow/.style={->, thick},
    node distance=0.55cm and 0cm,
]

\node[block, fill=blue!8]   (S1)
    {\textbf{Step 0.} Cart-pole dynamics $\dot{\mathbf{x}}=\mathbf{f}(\mathbf{x},u)$
     \eqref{eq:eom}; identify the upright equilibrium
     $\mathbf{x}^*=[0,0,\pi,0]^\top$};

\node[block, fill=teal!10, below=of S1] (S2)
    {\textbf{Step 1.} Energy error
     $\energyError=\tfrac{1}{2}m\ell^2\dot\theta^2 - mg\ell(1+\cos\theta)$
     \eqref{eq:energy_error}; the target $(\pi,0)$ lies on the homoclinic
     orbit $\mathcal{S}_0$ \eqref{eq:homoclinic}};

\node[block, fill=teal!10, below=of S2] (S3)
    {\textbf{Step 2.} Augmented Lyapunov candidate
     $V_{\mathrm{aug}}=\tfrac{1}{2}\energyError^2+\tfrac{\mu}{2}\dot{x}^2$
     \eqref{eq:Lyapunov_function_aug}; derive $F_{\mathrm{aug}}$
     \eqref{eq:swing_control_law_aug} so that $\dot V_{\mathrm{aug}}\le 0$}; 

\node[block, fill=teal!10, below=of S3] (S4)
    {\textbf{Step 3.} LaSalle principle: pendulum trajectories converge to $\mathcal{S}_0$ with
     $\dot{x}\to 0$ (Thm.~\ref{thm:swingup_aug}); $(\theta,\dot\theta)$ enters a neighborhood of $(\pi,0)$ in finite time (Cor.~\ref{cor:arrival})};

\node[decision, below=0.6cm of S4] (D1)
    {$\mathbf{z}\in\Sigma$? \eqref{eq:switching_region}, $\mathbf{z} \triangleq \mathbf{x} - \mathbf{x}^*$};

\node[block, fill=orange!10, right=0.7cm of D1, text width=2.4cm] (cont)
    {Continue swing-up: $u = F_{\mathrm{aug}}$
     \eqref{eq:swing_control_law_aug}};

\node[block, fill=purple!10, below=0.6cm of D1] (S5)
    {\textbf{Step 4.} Engage LQR \eqref{eq:LQR_control_law}: $u=-K\mathbf{z}$;
     handoff valid since $\Sigma\subseteq\Omega_{c^*}$ (Thm.~\ref{thm:main})};

\node[block, fill=purple!10, below=of S5] (S6)
    {\textbf{Step 5.} $W=\mathbf{z}^\top P_{\mathrm{lqr}}\mathbf{z}$;
     $\dot W_{\mathrm{nl}}<0$ and $\Omega_{c^*}$ forward-invariant
     (Lem.~\ref{lem:lqr})};

\node[block, fill=green!8, below=of S6] (S7)
    {\textbf{Result.}
     $[x,\dot x,\theta,\dot\theta]^\top\to[0,0,\pi,0]^\top$ for every
     $\mathbf{x}(0)$ in a compact set about $\mathbf{0}$, not on the stable
     manifold of $\mathcal{M}_B = \{(x,0,0,0) \:|\: x \in \mathbb{R}\}$};

\draw[arrow] (S1) -- (S2);
\draw[arrow] (S2) -- (S3);
\draw[arrow] (S3) -- (S4);
\draw[arrow] (S4) -- (D1);
\draw[arrow] (D1.east) -- node[above, font=\footnotesize]{No} (cont.west);
\draw[arrow] (cont.east) -- ++(0.35,0) |- (S2.east);
\draw[arrow] (D1) -- node[right, font=\footnotesize]{Yes} (S5);
\draw[arrow] (S5) -- (S6);
\draw[arrow] (S6) -- (S7);

\end{tikzpicture}
\caption{Overall control architecture. Steps~1--3 (swing-up phase) drive the pendulum to the homoclinic orbit $\mathcal{S}_0$ with the cart at rest and carry the state into the switching region $\Sigma$ in finite time;
decision block (diamond) is the one-way switching condition; Steps~4--5 (stabilization phase) hand off to the LQR, whose certified region of attraction $\Omega_{c^*}$ contains $\Sigma$. The swing-up phase continues until $\mathbf{z}\in\Sigma$.}
\label{fig:flowchart}
\end{figure}

\begin{figure*}[t]
    \centering
    \begin{subfigure}[b]{0.495\textwidth}
        \centering
        \includegraphics[trim={5cm 1.5cm 4.5cm 2cm}, clip=true, width=\linewidth]{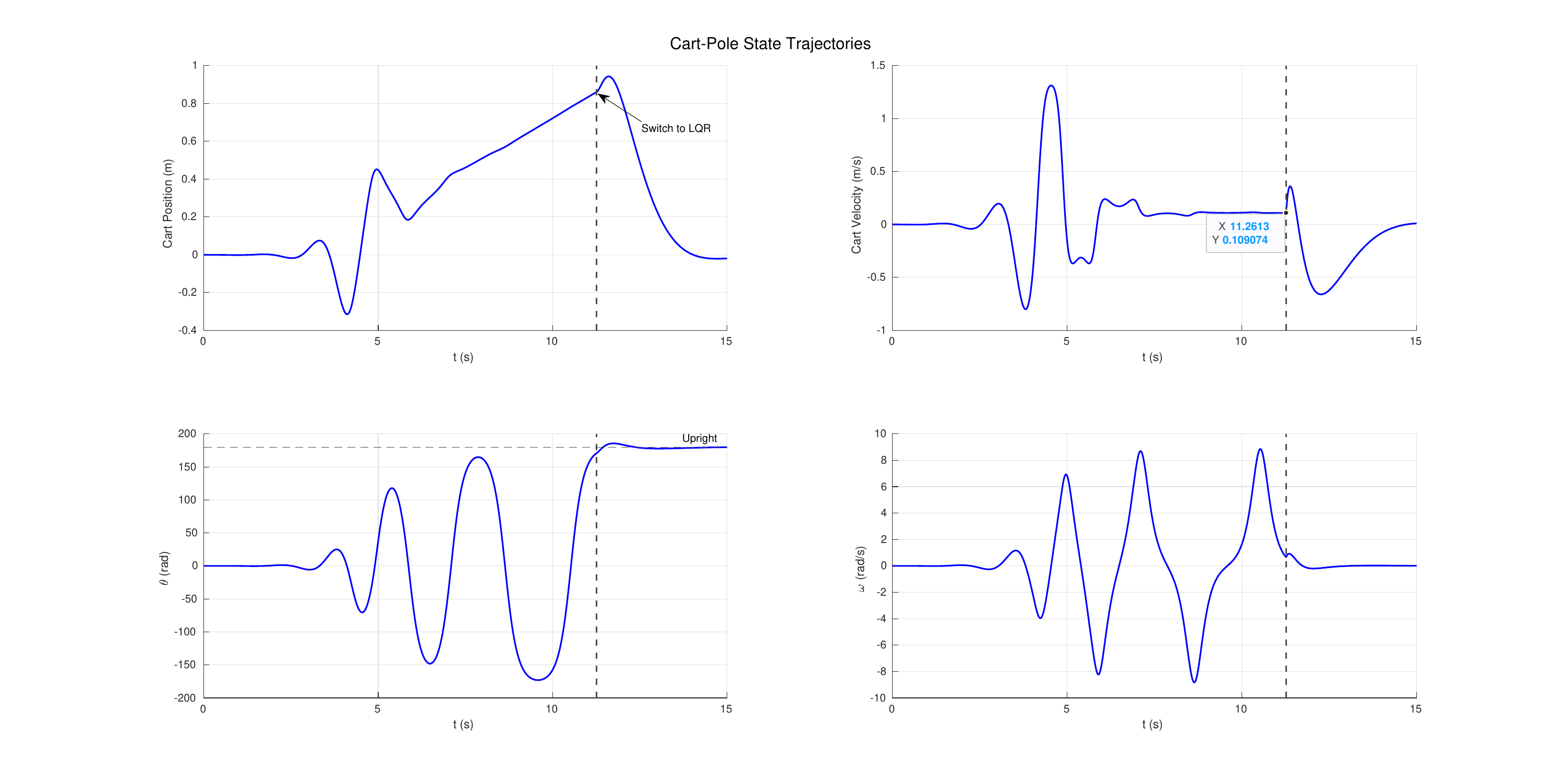}
        \caption{Control law \eqref{eq:swing_control_law} derived from \textit{unaugmented} $V$ \eqref{eq:Lyapunov_function}}
        \label{fig:unaugmented_states}
    \end{subfigure}
    \hfill 
    \begin{subfigure}[b]{0.495\textwidth}
        \centering
        \includegraphics[trim={5cm 1.5cm 4.5cm 2cm}, clip=true, width=\linewidth]{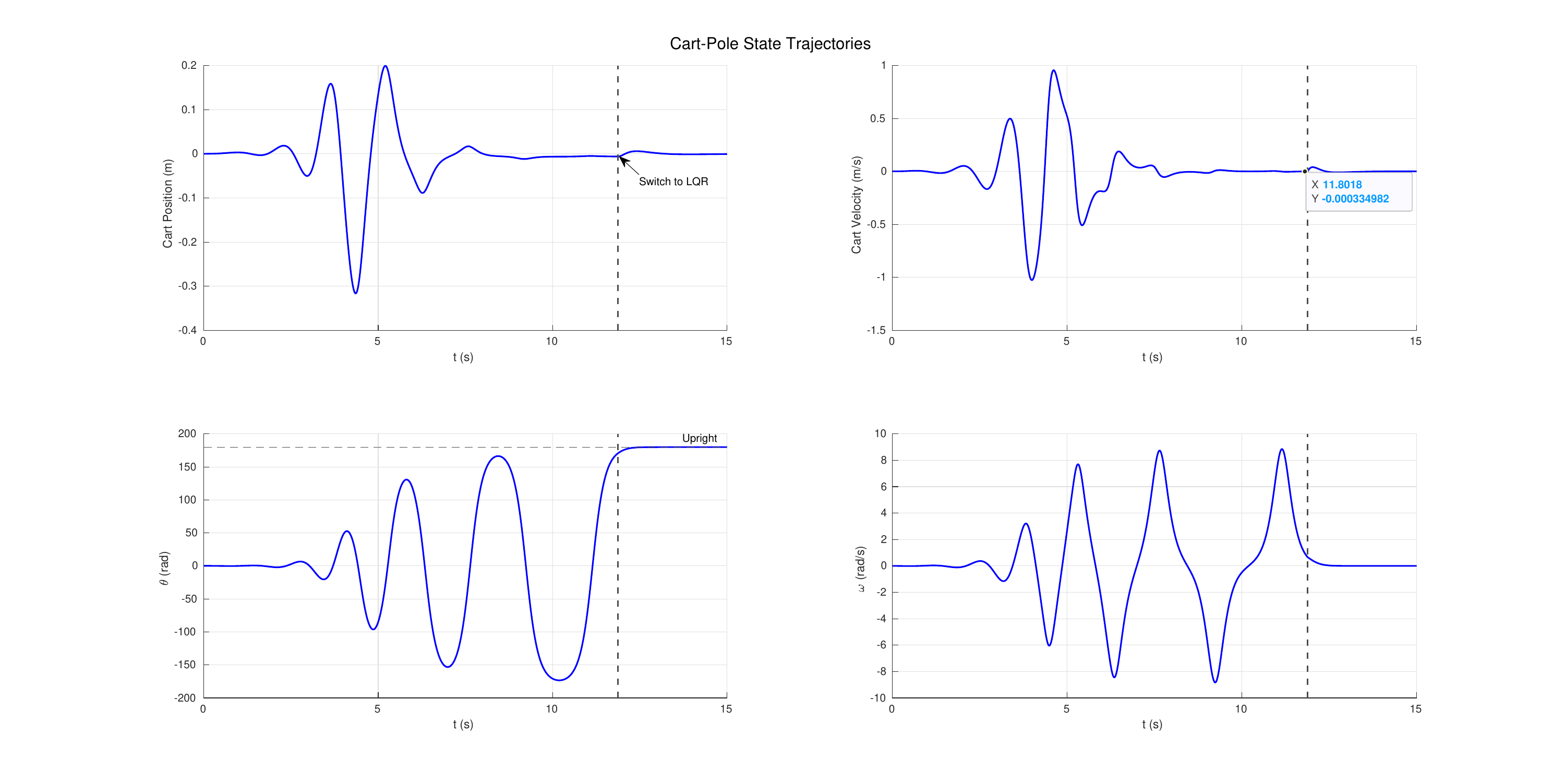}
        \caption{Control law \eqref{eq:swing_control_law_aug} derived from \textit{augmented} $V$ \eqref{eq:Lyapunov_function_aug}}
        \label{fig:augmented_states}
    \end{subfigure}
    \caption{Simulated state trajectories of cart-pole swing-up and stabilization starting from $\mathbf{x}(0) = [0, 0, 0, 0]^\top$, under two different control laws. Note that in the left plot -- (a), cart velocity converges to a constant (non-zero) value, and in the right plot -- (b), cart velocity converges to 0, by the design of control law \eqref{eq:swing_control_law_aug} that includes a velocity-damping component. The switch time $T_{\mathrm{sw}}$ in both cases is empirically observed to be $\sim12$\,s.}
    \label{fig:swing-up_states}
\end{figure*}

\begin{figure*}[t]
    \centering
    \begin{subfigure}[b]{0.4\textwidth}
        \centering
        \includegraphics[trim={2.25cm 0cm 2.25cm 0.7cm}, clip=true, width=0.95\linewidth]{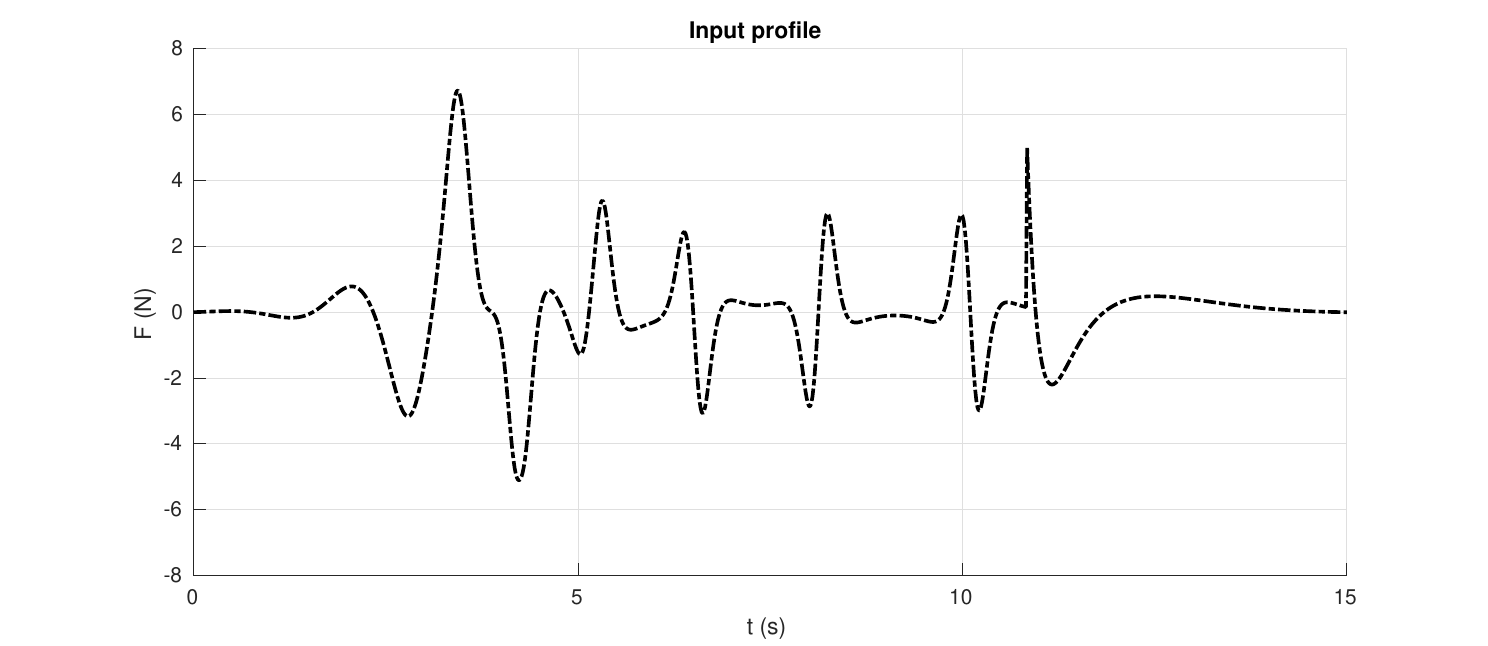}
        \caption{Control law \eqref{eq:swing_control_law} derived from \textit{unaugmented} $V$ \eqref{eq:Lyapunov_function}}
        \label{fig:unaugmented_inputprofile}
    \end{subfigure}
    \hspace{2em}
    \begin{subfigure}[b]{0.4\textwidth}
        \centering
        \includegraphics[trim={2.25cm 0cm 2.25cm 0.7cm}, clip=true, width=0.95\linewidth]{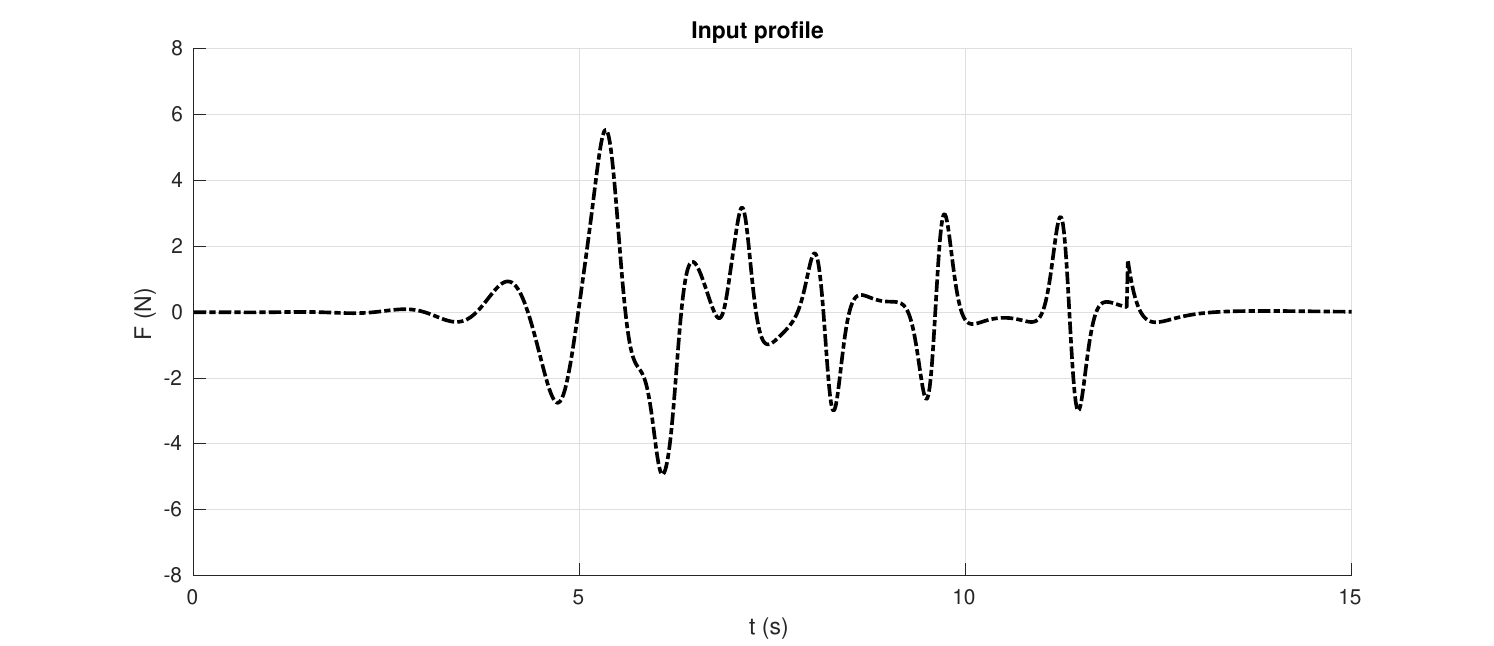}
        \caption{Control law \eqref{eq:swing_control_law_aug} derived from \textit{augmented} $V$ \eqref{eq:Lyapunov_function_aug}}
        \label{fig:augmented_inputprofile}
    \end{subfigure}
    \caption{Sample input profiles for cart-pole swing-up and stabilization.}
    \label{fig:swing-up_inputs}
\end{figure*}

\begin{theorem}[\textsc{Certified swing-up and stabilization}]\label{thm:main}
Let $P_{\mathrm{lqr}}$ solve \eqref{eq:are} and $c^*$ be as in \eqref{eq:cstar}. Choose gains $k,\mu>0$ and thresholds
$\delta_x,\delta_v,\delta_\phi,\delta_\omega>0$ such that
$\Sigma\subseteq\Omega_{c^*}$. Then for every initial condition in a compact set about the downward equilibrium, excluding the stable manifold of $\mathcal{M}_B = \{(x,0,0,0) \:|\: x \in \mathbb{R}\}$, the switching controller \eqref{eq:switching_controller}
drives the state to the upright equilibrium:
\begin{equation}
  [x(t),\dot{x}(t),\theta(t),\dot{\theta}(t)]^\top
    \longrightarrow [0,0,\pi,0]^\top.
  \label{eq:main_conv}
\end{equation}
\end{theorem}

\begin{proof}

\emph{Phase 1 (swing-up), $\mathbf{z}\notin\Sigma$}: Under
control law~\eqref{eq:swing_control_law_aug}, Theorem~\ref{thm:swingup_aug} gives
$\dot{x}\to 0$ and $(\theta,\dot{\theta})\to\mathcal{S}_0$ for every initial condition not on the stable manifold of $\mathcal{M}_B$, and Lemma~\ref{lem:cart_drift} bounds the cart position. By Corollary~\ref{cor:arrival}, the pendulum enters any neighborhood of $(\pi,0)$ in finite time; together with $\dot{x}\to 0$ and the bounded cart position, the system state enters $\Sigma$ at a finite switching time
$T_{\mathrm{sw}}<\infty$.

\emph{Handoff at $t=T_{\mathrm{sw}}$}: By the design condition
$\Sigma\subseteq\Omega_{c^*}$, the switching state satisfies
$\mathbf{z}(T_{\mathrm{sw}})\in\Sigma\subseteq\Omega_{c^*}$.

\emph{Phase 2 (stabilization), $\mathbf{z}\in\Sigma$}: By
Lemma~\ref{lem:lqr}, $\Omega_{c^*}$ is forward-invariant under LQR law $u=-K\mathbf{z}$
and $\mathbf{z}(t)\to\mathbf{0}$, i.e.
$[x,\dot{x},\theta,\dot{\theta}]^\top\to[0,0,\pi,0]^\top$.
\end{proof}

\begin{remark}
The containment condition $\Sigma\subseteq\Omega_{c^*}$ enters Theorem~\ref{thm:main} as a design requirement on the thresholds rather than a consequence of the swing-up bounds --- the analytic cart-position bound $\bar{x}$ of Lemma~\ref{lem:cart_drift} is too loose to certify containment directly. In practice we verify $\mathbf{z}(T_{\mathrm{sw}})\in\Omega_{c^*}$ numerically on the realized trajectory (Section~\ref{sec:results}). Furthermore, the augmented law's velocity regulation ($\dot{x}(T_{\mathrm{sw}})\approx 0$) aids in placing the switching state well inside $\Omega_{c^*}$.
\end{remark}

\section{Numerical Simulations}
\label{sec:results}

For cart-pole parameters $M=1.0$\,kg, $m=0.1$\,kg, $\ell=0.5$\,m, $g=9.81$\,m\,s$^{-2}$, energy-control gains $k=2$, $\mu=1$, LQR cost matrices $Q=\mathrm{diag}(10,1,100,1)$, $R=0.5$, and initial conditions as small perturbations about $\mathbf{x}_0=[0,0,0,0]^\top$, the relevant quantities evaluate as follows.

\emph{Initial energy error:}
$|\energyError_0| = 2mg\ell = 0.981$\,J. 

\emph{Cart velocity bound \eqref{eq:cart_bounds}:}
$\bar{v} = {|\energyError_0|}/\sqrt{\mu} = 0.981$\,m\,s$^{-1}$, consistent with the peak cart velocity of $\approx 1.0$\,m\,s$^{-1}$ observed in simulation (see \Cref{fig:swing-up_states}).

\emph{Cart position bound \eqref{eq:cart_bounds}} with the empirically observed switch-time $T_{\mathrm{sw}}\approx 12$\,s: $\bar{x} = \bar{v}\,T_{\mathrm{sw}} \approx 11.772$\,m. This bound provides a valid reachability guarantee but is highly conservative (Remark~\ref{rmk:conservative_bound}): in simulation under the augmented law we observe $|x(t)| \lesssim 0.4$\,m for all $t \in [0, T_{\mathrm{sw}}]$
(\Cref{fig:augmented_states}).

\begin{figure*}[t]
    \centering
    \includegraphics[trim={3.25cm 0cm 3cm 0.2cm},clip=true,width=0.85\linewidth]{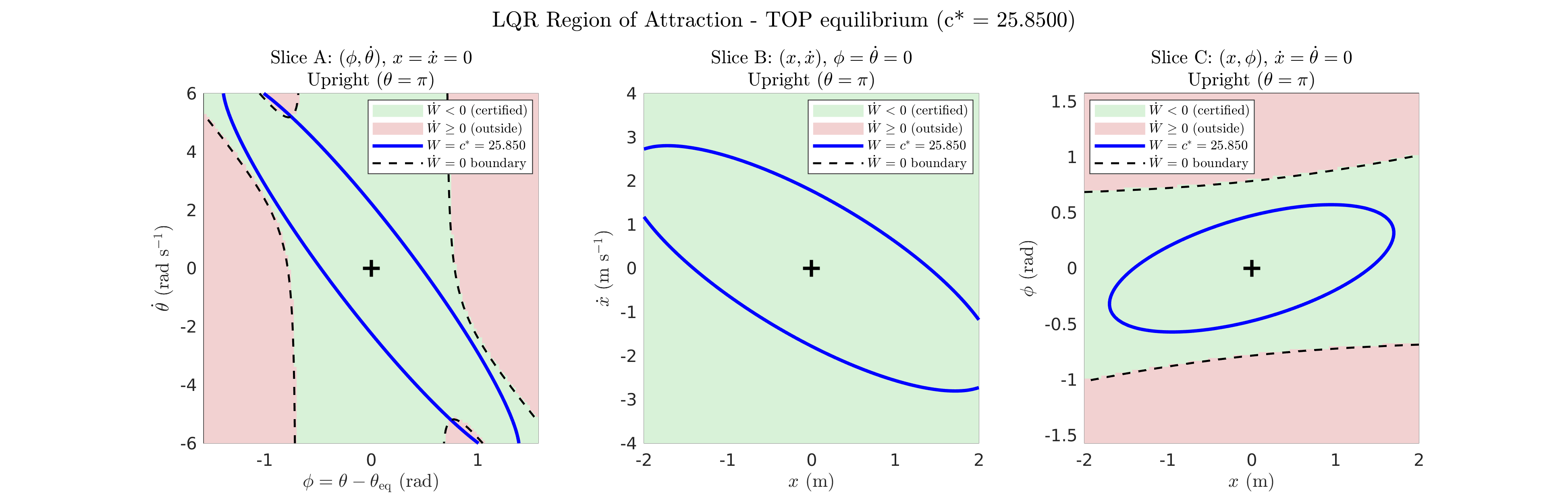}
    \caption{The ellipsoidal region-of-attraction certificate of the local
    LQR attractor synthesized at $\mathbf{x}_\mathrm{eq}=[0,0,\pi,0]^\top$, obtained from a radial line search. Recall from \eqref{eq:LQR_Lyapunov_function} that $W(\mathbf{z})=\mathbf{z}^{\top}P_{\mathrm{lqr}}\mathbf{z}$, where $\mathbf{z}=\mathbf{x}-\mathbf{x}_\mathrm{eq}$.}
    \label{fig:LQR_ROA}
\end{figure*}

\emph{LQR RoA (see \Cref{fig:LQR_ROA}):} A radial line search yields $c^* = 25.850$, with $\lambda_{\max}(P_{\mathrm{lqr}}) = 122.7$ for $P_{\mathrm{lqr}}$ obtained by solving \eqref{eq:are}. The maximum per-axis extents
$\sqrt{c^*[P_{\mathrm{lqr}}^{-1}]_{kk}}$ certify $|x|\leq 2.24$\,m, $|\dot{x}|\leq 4.35$\,m\,s$^{-1}$, $|\phi|\leq 1.44$\,rad, and $|\dot\phi|\leq 8.54$\,rad\,s$^{-1}$ within the RoA, $\Omega_{c^*}$ (\Cref{def:roa}).

\emph{Containment check:} The analytic cart-position bound $\bar{x}$ is far too loose to certify switching region $\Sigma\subseteq\Omega_{c^*}$ through the
Rayleigh--Ritz sufficient condition of Remark~\ref{rmk:containment}, exceeding the per-axis RoA extent of $2.24$\,m by nearly an order of magnitude. We therefore verify containment directly on the realized trajectory: evaluating $W(\mathbf{z}(T_{\mathrm{sw}}))$ from \eqref{eq:LQR_Lyapunov_function} yields $W(\mathbf{z}(T_{\mathrm{sw}}))\ll c^*$, confirming that the switching state lies well inside $\Omega_{c^*}$, thus validating the handoff. 

\subsection{Results}
Sample simulations are shown in
Figures~\ref{fig:swing-up_states}--\ref{fig:swing-up_inputs}. Under both control laws the energy-pumping phase is oscillatory: the pendulum accumulates amplitude over successive swings while the input force stays bounded within $\pm 8$\,N (\Cref{fig:swing-up_inputs}), and the state is handed to the LQR at $T_{\mathrm{sw}}\approx 12$\,s, after which it converges to the upright equilibrium $[0,0,\pi,0]^\top$. The two laws differ primarily in the steady-state cart velocity: the unaugmented law~\eqref{eq:swing_control_law} settles it to a non-zero constant (see \Cref{fig:unaugmented_states} and Remark~\ref{rmk:conservative_bound}), whereas the augmented law \eqref{eq:swing_control_law_aug} damps it to zero at switching (\Cref{fig:augmented_states}), corroborating the velocity-regulation design.

This difference carries over to the cart displacement off the origin. Under the unaugmented law the cart drifts up to $|x(t)|\lesssim 1$\,m (\Cref{fig:unaugmented_states}) before switching, whereas under the augmented law it stays within $|x(t)|\lesssim 0.4$\,m (\Cref{fig:augmented_states}) --- far inside both the conservative analytic bound $\bar{x}\approx 11.77$\,m and the LQR RoA's per-axis position extent of $2.24$\,m. The switching state therefore satisfies $W(\mathbf{z}(T_{\mathrm{sw}}))\ll c^*$, so the containment condition
$\Sigma\subseteq\Omega_{c^*}$ holds with a wide margin, validating the handoff between the energy-pumping swing-up controller and the stabilizing LQR.

\section{Conclusion}

We presented an end-to-end reachability analysis for cart-pole swing-up and stabilization, certifying that trajectories from a compact neighborhood of the downward equilibrium (excluding it) reach the upright equilibrium under a switched energy-based/LQR controller. The swing-up control design exploits the phase-space geometry of the conservative pendulum: the upright equilibrium lies on the homoclinic orbit separating oscillation from rotation, so an energy-shaping law that drives the energy error to zero steers the pendulum state onto this orbit. 
Convergence to the orbit follows from LaSalle's invariance principle; an augmented Lyapunov function additionally regulates the steady-state cart velocity to zero, for which we characterized the complete invariant set and proved instability of the downward equilibrium. A local LQR stabilizes the upright equilibrium, with a certified ellipsoidal region of attraction into which the swing-up phase delivers the state. Since the analytic cart-position bound is too loose to certify this containment directly, we verify it numerically on the realized trajectory, and simulations corroborate the theoretical analysis.


\section*{Acknowledgments}
The author thanks his advisor, Prof.~Michael Otte, for his guidance
throughout this work, and Prof.~Nikhil Chopra for insightful discussions that helped shape these ideas.



\clearpage
\appendices
\section{Augmented Lyapunov Function \eqref{eq:Lyapunov_function_aug}: Derivative and Invariant Set Characterization}
\label{app:aug}

This appendix derives $\dot{V}_{\mathrm{aug}}$, characterizes the invariant
set of the closed loop under augmented control law~\eqref{eq:swing_control_law_aug}, and
establishes instability of the downward equilibrium set $\mathcal{M}_B$.

\subsection{Time derivative of $V_{\mathrm{aug}}$}
Recall $V_{\mathrm{aug}} = \tfrac{1}{2}\energyError^2 +
\tfrac{\mu}{2}\dot{x}^2$. Differentiating along \eqref{eq:eom},
\begin{equation}
  \dot{V}_{\mathrm{aug}} = \energyError\dot{E} + \mu\dot{x}\ddot{x}.
  \label{eq:app_Vdot}
\end{equation}
From \eqref{eq:Edot_full},
\begin{equation}
  \energyError\dot{E}
    = \frac{-m\ell\dot{\theta}c_\theta\,\energyError}{\denom}\,F
    - \frac{m^2\ell\dot{\theta}s_\theta c_\theta
            (\ell\dot{\theta}^2+gc_\theta)}{\denom}\,\energyError,
  \label{eq:app_EEdot}
\end{equation}
and substituting $\ddot{x}$ from \eqref{eq:eom},
\begin{equation}
  \mu\dot{x}\ddot{x}
    = \frac{\mu\dot{x}}{\denom}\,F
    + \frac{\mu m\dot{x}s_\theta(\ell\dot{\theta}^2+gc_\theta)}{\denom}.
  \label{eq:app_xxdot}
\end{equation}
Adding \eqref{eq:app_EEdot} and \eqref{eq:app_xxdot} and collecting the
coefficient of $F$ and the autonomous remainder,
\begin{align}
\begin{split}
  \dot{V}_{\mathrm{aug}}
    = \: &\frac{1}{\denom}
      \underbrace{\bigl[\mu\dot{x} - m\ell\dot{\theta}c_\theta\energyError\bigr]}_{\sigma}
      F \\
   & + \frac{(\ell\dot{\theta}^2+gc_\theta)}{\denom}
      \underbrace{\bigl[\mu m\dot{x}s_\theta
        - m^2\ell\dot{\theta}s_\theta c_\theta\energyError\bigr]}_{m s_\theta\sigma}.
  \label{eq:app_factor}
\end{split}
\end{align}
With $\sigma \triangleq \mu\dot{x} - m\ell\dot{\theta}c_\theta\energyError$,
the autonomous bracket factors as $m s_\theta\sigma$, giving
\begin{equation}
  \dot{V}_{\mathrm{aug}}
    = \frac{\sigma}{\denom}\bigl[F + m s_\theta(\ell\dot{\theta}^2+gc_\theta)\bigr].
  \label{eq:app_key}
\end{equation}
Choosing $F$ so the bracket equals $-k\sigma$ yields
\eqref{eq:swing_control_law_aug}, hence
$\dot{V}_{\mathrm{aug}} = -k\sigma^2/\denom \leq 0$.

\subsection{Invariant set characterization}
By LaSalle, trajectories converge to the largest invariant subset of
$\mathcal{Z} = \{\sigma = 0\} = \{\mu\dot{x} = m\ell\dot{\theta}c_\theta\energyError\}$.
On any trajectory with $\sigma\equiv 0$, \eqref{eq:swing_control_law_aug}
reduces to $F = -m s_\theta(\ell\dot{\theta}^2+gc_\theta)$. Substituting
into \eqref{eq:eom} gives $\ddot{x}=0$ (so $\dot{x}$ is constant on
$\mathcal{Z}$) and $\dot{E}=0$ (so $\energyError$ is constant on
$\mathcal{Z}$; call it $\energyError_0$). Differentiating $\sigma$ with
$\ddot{x}=0$ and $\dot{\energyError}=0$,
\begin{equation}
  \dot{\sigma}\big|_{\sigma=0}
    = -m\ell\energyError_0(\ddot{\theta}c_\theta - \dot{\theta}^2 s_\theta) = 0,
  \label{eq:app_sigmadot}
\end{equation}
requiring either $\energyError_0 = 0$ or
$\ddot{\theta}c_\theta = \dot{\theta}^2 s_\theta$.

\emph{Case A} ($\energyError_0 = 0$): $\sigma=0$ then forces $\dot{x}=0$,
and the pendulum lies on $\mathcal{S}_0$ under conservative dynamics, on
which the entire orbit is invariant. The corresponding invariant set is
therefore
\begin{equation}
  \mathcal{M}_A = \bigl\{(x,\dot{x},\theta,\dot{\theta}) \:\big|\:
    \dot{x}=0,\ (\theta,\dot{\theta})\in\mathcal{S}_0 \bigr\},
  \label{eq:MA}
\end{equation}
i.e. the cart at rest with the pendulum on the homoclinic orbit. The
swing-up law delivers the state to $\mathcal{M}_A$; convergence
\emph{along} $\mathcal{S}_0$ to the upright equilibrium is completed by the LQR phase (Section~\ref{subsec:switching}).

\emph{Case B:} ($\energyError_0\neq 0$): evaluating $\ddot{\theta}$ on
$\mathcal{Z}$ using \eqref{eq:eom} with
$F=-m s_\theta(\ell\dot{\theta}^2+gc_\theta)$,
\begin{equation}
  \ddot{\theta}\big|_{\sigma=0}
    = \frac{-gs_\theta[(M+m)-mc_\theta^2]}{\ell\denom}
    = -\frac{g s_\theta}{\ell},
  \label{eq:app_thetaddot}
\end{equation}
using $(M+m)-mc_\theta^2 = \denom$. Substituting into
$\ddot{\theta}c_\theta = \dot{\theta}^2 s_\theta$ gives
$s_\theta = 0$ or $\dot{\theta}^2 = -gc_\theta/\ell$.

\emph{Sub-case B1:} ($s_\theta=0$, i.e. $\theta\in\{0,\pi\}$): invariance
forces $\dot{\theta}=0$. At $\theta=\pi$, $\energyError=0$, contradicting
$\energyError_0\neq 0$. At $\theta=0$, $\energyError=-2mg\ell\neq 0$
(consistent) and $\ddot{\theta}=0$; moreover $\sigma=0$ evaluates to
$\mu\dot{x} - m\ell(0)(1)(-2mg\ell) = \mu\dot{x} = 0$, so $\dot{x}=0$.
This gives the downward equilibrium set $\mathcal{M}_B = \{(x,0,0,0)\}$,
with the cart velocity pinned to zero by the augmentation.

\emph{Sub-case B2:} ($\dot{\theta}^2 = -gc_\theta/\ell$, $s_\theta\neq 0$):
invariance requires $\tfrac{d}{dt}(\dot{\theta}^2 + gc_\theta/\ell)=0$,
i.e. $\dot{\theta}(2\ddot{\theta} - gs_\theta/\ell)=0$; for
$\dot{\theta}\neq 0$ this needs $\ddot{\theta}=gs_\theta/(2\ell)$, which
contradicts \eqref{eq:app_thetaddot} unless $s_\theta=0$. No new invariant trajectories arise.

Collecting the two cases, invariant set
$\mathcal{M} = \mathcal{M}_A \cup \mathcal{M}_B$.

\subsection{Instability of $\mathcal{M}_B$}
Linearizing the system~\eqref{eq:eom} under control law~\eqref{eq:swing_control_law_aug} about $(x,0,0,0)$, with
$s_\theta\approx\theta$, $c_\theta\approx 1$, $\denom\approx M$, and $\energyError\approx -2mg\ell$, the Jacobian is
\begin{equation}
  A_\downarrow = \begin{bmatrix}
    0 & 1 & 0 & 0 \\
    0 & -k\mu & 0 & -2km^2g\ell^2 \\
    0 & 0 & 0 & 1 \\
    0 & k\mu/\ell & -g/\ell & 2km^2g\ell
  \end{bmatrix}.
  \label{eq:app_Adown}
\end{equation}
The characteristic polynomial of the $(\theta,\dot{\theta})$ block
(setting $\dot{x}=0$) is
$\lambda^2 - 2km^2g\ell\,\lambda + g/\ell = 0$, whose roots have product
$g/\ell > 0$ and sum $2km^2g\ell > 0$; both eigenvalues have positive real
part. Hence $\mathcal{M}_B$ is unstable: any perturbation in $\theta$ or
$\dot{\theta}$ drives the trajectory away from the downward equilibrium.

\end{document}